\documentclass[11pt,reqno]{amsart}
\usepackage{amsmath,amssymb, mathrsfs, amsfonts}
\usepackage{graphicx, color, amscd, bm, slashed}
\usepackage{comment}
\usepackage{slashed}
\newcommand{\RN}[1]{%
  \textup{\uppercase\expandafter{\romannumeral#1}}%
}
\newcommand{\pl}{\partial}

\newcommand{\rw}{\rightarrow}

\newcommand{\Lb}{\underline{L}}
\newcommand{\chib}{\underline{\chi}}

\newcommand{\R}{\mathbb{R}}

\newcommand{\lt}{\left}
\newcommand{\rt}{\right}
\newcommand{\na}{\nabla}
\renewcommand{\hat}{\widehat}

\newcommand{\tr}{\mbox{tr}}

\newtheorem{theorem}{Theorem}

\newtheorem{prop}[theorem]{Proposition}

\theoremstyle{definition}
\newtheorem{definition}[theorem]{Definition}
\newtheorem{remark}[theorem]{Remark}

\title[Supertranslation invariance in double null gauge]{Supertranslation invariance of angular momentum at null infinity in double null gauge}
\author[P.-N. Chen, M.-T. Wang, Y.-K. Wang, and S.-T. Yau]{Po-Ning Chen, Mu-Tao Wang,\\ Ye-Kai Wang, and Shing-Tung Yau}
\numberwithin{theorem}{section}
\numberwithin{equation}{section}

\begin{document}

\begin{abstract} 
The supertranslation invariance of the Chen-Wang-Yau (CWY) angular momentum in the Bondi-Sachs formalism/gauge was ascertained by the authors in \cite{CKWWY_evol, CWWY_atmp}. In this article, we study the corresponding problem in the double null gauge. In particular, supertranslation ambiguity of this gauge is identified and the CWY angular momentum is proven to be free of this ambiguity. A similar result is obtained for the CWY center of mass integral. 
\end{abstract}

\thanks{P.-N. Chen is supported by NSF grant DMS-1308164 and Simons Foundation collaboration grant \#584785, M.-T. Wang is supported by NSF grants DMS-1810856 and DMS-2104212, Y.-K. Wang is supported by MOST Taiwan grant 107-2115-M-006-001-MY2 and 109-2628-M-006-001 -MY3, and S.-T. Yau is supported by NSF grants  PHY-0714648 and DMS-1308244.} 

\maketitle

\section{}
The friendship of Professor Christodoulou and Yau started in 1981. In the first months of his commute from Syracuse to the Institute of Advanced Study, Christodoulou stayed in Yau's apartment. He also sat in Yau's course on minimal surface theory and applications for two years \cite{C_reminiscences}. Christodoulou eventually applied geometric analysis to hyperbolic partial differential equations, culminating first in the monumental proof of stability of the Minkowski spacetime with Klainerman \cite{CK} and later in the formation of black holes in general relativity \cite{C_formation} and the formation of shocks in fluid mechanics \cite{C_shock}. We take this opportunity to celebrate Christodoulou's achievement and the friendship that has lasted more than 40 years.

In a series of papers in the 1960's, Bondi \cite{Bondi}, Bondi-van der Burg-Metzner \cite{BVM}, and Sachs \cite{S} provided one of the first convincing theoretical evidences for gravitational radiation. Assuming an asymptotically Minkowskian isolated system admits a coordinate system adapted to an optical function $u$ and a luminosity distance $r$, Bondi et. al. were able to solve the vacuum Einstein equation in power series of $r^{-1}$, defined the energy for the $u=$const. null hypersurfaces
\[ E(u) = \frac{1}{4\pi} \int_{S^2} m, \]
and derived the {\it Bondi mass loss formula}\footnote{The function $m$ and 2-tensor $N_{AB}$ appear in the Taylor expansion of metric coefficients. See Section 4.1 for their definition.}
\[ \pl_u E = -\frac{1}{32\pi} \int_{S^2} N_{AB}N^{AB} \le 0, \]   
which is interpreted as saying an isolated system emits gravitational waves that carries mass away from the system.

The Bondi-Sachs coordinate-based approach (henceforth referred to as Bondi-Sachs formalism \cite{MW}) was incorporated into Penrose's conformal treatment of null infinity. However, it is difficult to construct a large class of spacetimes that admit conformal compactifications (or Bondi-Sachs coordinates). It is in the work of stability of Minkowski spacetime \cite{CK} that Christodoulou-Klainerman gives a rigorous treatment of null infinity in the setting of initial value problems: a detailed asymptotic decay estimate for Ricci coefficients and curvature components are derived, the Bondi mass is defined by the limit of Hawking mass, and the mass loss formula is rigorously proved. The work also leads to the discovery of the Christodoulou memory effect \cite{Christo1991}.

Defining angular momentum for gravitational fields turns out to be more subtle. In special relativity, associated with a particle (represented by a curve $\gamma$ in $\R^{3,1}$) are the conserved quantities
\begin{enumerate}
\item energy $E= \langle \gamma', \frac{\pl}{\pl t} \rangle$
\item linear momentum $P^i = \langle \gamma', \frac{\pl}{\pl x^i} \rangle$
\item angular momentum $ J_{ij}= \langle \gamma', x^i \frac{\pl}{\pl x^j} - x^j \frac{\pl}{\pl x^i} \rangle$
\item center of mass $C^i = \langle \gamma', t \frac{\pl}{\pl x^i} + x^i \frac{\pl}{\pl t}\rangle$
\end{enumerate}
If one changes the coordinate system by a translation $\bar t = t + \alpha^0, \bar x^i= x^i + \alpha^i$, energy and linear momentum remain the same but the angular momentum and center of mass transform as
\begin{align}\label{SR}
\begin{split}
\bar J_{ij} &= J_{ij} + \alpha^i P^j - \alpha^j P^i, \\
\bar C^i &= C^i + \alpha^0 P^i + \alpha^i E.
\end{split}
\end{align}
We conclude that the angular momentum and center of mass have a ``4-dimensional translation ambiguity" that comes from the choice of the origin.

In Bondi-Sachs formalism, the symmetry group of null infinity is given by coordinate transformations that preserve the form of metric tensors and asymptotics. Although structurally similar to the Poincar\'e group, it is the semi-direct product of the Lorentz group and an infinite-dimensional abelian group---{\it supertranslations}---instead of the 4-dimensional group of translations. The transformation law of all existing proposals of angular momentum under supertranslations bears no resemblance to \eqref{SR} and physicists have yet to find an interpretation. This is the ``supertranslation ambiguity" that has puzzled researchers in the field since the 1960s.

The purpose of this article is to explain the supertranslation invariance of Chen-Wang-Yau angular momentum, and center-of-mass integral \cite{CKWWY_evol, CWWY_atmp} with the main results presented in the double null gauge rather than the Bondi-Sachs formalism as in \cite{CKWWY_evol, CWWY_atmp}.

In Section 2, the definitions of the Chen-Wang-Yau quasi-local angular momentum and center-of-mass integral are given. They are extensions of the Wang-Yau quasi-local mass \cite{Wang-Yau1, Wang-Yau2} and defined for spacelike 2-surfaces in spacetimes. The key idea is to isometrically embed the surface into the Minkowski spacetime and then pull back the Killing vector fields along the surface. This resolves two difficulties encountered in the Hamiltonian approach to define conserved quantities in general relativity: the lack of background coordinate systems and the absence of symmetries. We take the limits, as $r \rw \infty$ of quasi-local Chen-Wang-Yau angular momentum and center-of-mass on $r=$const. surfaces to get the corresponding Chen-Wang-Yau quantities for each $u=$const. null hypersurface. 

In Section 3.1, we present A. Rizzi's definition of angular momentum at null infinity. It is defined in the framework of stability of Minkowski spacetime and is proposed in Rizzi's thesis supervised by Christodoulou. In Section 3.2, we discuss a conservation law of angular momentum that is related to Christodoulou's conservation law of linear momentum \cite{Christo1991}.      

The Bondi-Sachs formalism is recalled in Section 4.1, and supertranslations are defined. For the rest of the article, we focus on the double null gauge, which is familiar to the readers of Christodoulou's work. In Section 4.2, we define supertranslations in this context, compute the transformations of Ricci coefficients, and discuss Rizzi's attempt to restore the 4-dimensional dependence of origin. The supertranslation ambiguity is explained for the total flux of Rizzi's angular momentum, which is equal to the (classical) Ashtekar-Streubel definition \cite{AS} in Section 4.3. In the final Section 4.4, we show that the Chen-Wang-Yau angular momentum and center-of-mass integral transform according to \eqref{SR} under supertranslations, which is the main result in \cite{CKWWY_evol, CWWY_atmp} and is presented here in the double null gauge. 
  
\begin{definition}
Throughout this article, we denote the standard metric on $S^2$ by $\sigma$. More precisely, $\sigma = d\theta^2 + \sin^2\theta d\phi^2$ in spherical coordinates. We raise and lower indices with respect to $\sigma$. Let $\na$ and $\Delta$ be the covariant derivative and Laplace operator with respect to $\sigma$. If the volume form is omitted in the integrals $\int_{S^2} F$, it is taken with respect to the volume form of $\sigma$. We also need the spherical harmonics decomposition. Let $\mathcal{H}_{\ell\le 1}$ denote the space of functions spanned by $1, \tilde X^1 = \sin\theta\cos\phi, \tilde X^2 = \sin\theta\sin\phi, \tilde X^3 = \cos\theta$ and $\mathcal{H}_{\ell\ge 2}$ denote the space of functions supported in $\ell\ge 2$ modes.
\end{definition}

\section{Chen-Wang-Yau quasi-local angular momentum and center of mass}
Given a spacelike 2-surface $\Sigma$ in the spacetime, we get data $(\slashed{g}, |H|, \alpha_H)$ where $\slashed{g}$ is the induced metric, $|H|$ the norm of the mean curvature vector (which is assumed to be spacelike), $\alpha_H$ the connection 1-form of the normal bundle of $\Sigma$. In terms of Ricci coefficnets \cite[page 1486]{Christo1991},  
\begin{align*}
|H| = \sqrt{-\tr\chi \tr\chib}, \quad \alpha_H = \zeta - \frac{1}{2} d\log \tr\chi + \frac{1}{2} d \log (- \tr\chib)  
\end{align*}
are independent of the scaling $L \rw a L, \Lb \rw a^{-1} \Lb$ of the two null normal vector fields along $\Sigma$.

Consider an isometric embedding $X:(\Sigma, \slashed{g}) \rw \R^{3,1}$ into the Minkowski spacetime. For the image surface $X(\Sigma)$, we can compute the norm of the mean curvature vector $H_0$ and connection 1-form $\alpha_{H_0}$. For a constant timelike unit vector field $T_0$, viewed as an observer, let $\tau = - \langle X, T_0 \rangle$. Define the mass density function and momentum density one-form by
\begin{align*}
\bm\rho &= \frac{\sqrt{|H_0|^2 + \frac{(\Delta \tau)^2}{1+|\nabla \tau|^2}} - \sqrt{|H|^2 + \frac{(\Delta \tau)^2}{1+|\nabla \tau|^2}}}{\sqrt{1+|\nabla\tau|^2}} \\
\bm j &= - \bm\rho\nabla \tau + \nabla\left[\sinh^{-1}\left(\frac{\bm\rho\Delta\tau}{|H_0||H|}\right)\right] +\alpha_{H_0} - \alpha_{H}
\end{align*} 
Note that the sign convention of $\bm j$ is opposite to \cite{CWY_ATMP, KWY} but coincides with \cite[Proposition 7.1]{CWY_conserved}.

There are many such choices of isometric embedding $X$ and observer $T_0$ in the Minkowski spacetime. We consider only those pairs $(X, T_0)$ with associated data satisfying the optimal isometric embedding equation
\begin{equation}\label{oiee}
\slashed{div} \bm j = 0.
\end{equation}

We write $K_{ij} = x^i \frac{\pl}{\pl x^j} - x^j \frac{\pl}{\pl x^i}$ for the rotation Killing vector fields and $K_i = t \frac{\pl}{\pl x^i} + x^i \frac{\pl}{\pl t}$ for the boost Killing vector fields; together they form a basis of the Lorentz algebra. Here is the definition of Chen-Wang-Yau quasi-local conserved quantities \cite[Definition 2.2]{CWY_conserved}.
\begin{definition}\label{qlCharge}
The quasi-local conserved quantity of $\Sigma$ with respect to an optimal isometric embedding $(X,T_0)$ and a Killing field $K$ in Minkowski spacetime is
\begin{align}
E(\Sigma, X,T_0, K) = \frac{1}{8\pi}\int_\Sigma - \langle K,T_0 \rangle \bm\rho + \bm j(K^\top) d\Sigma 
\end{align}
where $K^\top$ denotes the projection of a Lorentz Killing field $K$ onto the tangent space of the image $X(\Sigma)$.

Suppose $T_0 = A \lt( \frac{\pl}{\pl t} \rt)$ for a Lorentz transformation, then the quasi-local conserved quantities corresponding to $A(K_{ij})$ are called the quasi-local angular momentum with respect to $T_0$ and the quasi-local conserved quantities corresponding to $A(K_i)$ are called the quasi-local center of mass integrals with respect to $T_0$.
\end{definition}

Chen-Wang-Yau conserved quantities satisfy the following properties: 
\begin{enumerate}
\item They vanish for any spacelike 2-surface in Minkowski spacetime.
\item The quasi-local angular momentum coincides with the Komar integral for an axially symmetric 2-surface in an axially symmetric spacetime.
\end{enumerate}

Given an asymptotically flat initial data set $(M,g,k)$, the limits of Chen-Wang-Yau conserved quantities of the coordinate spheres $\Sigma_r$ give rise to the {\it total Chen-Wang-Yau conserved quantities}. On the one hand, the total energy-momentum recovers the ADM energy-momentum 4-vector \cite{WY_spatial}. On the other hand, the total angular momentum and center of mass integrals agree with ADM angular momentum and Regge-Teitelboim center of mass for harmonic asymptotes when the linear momentum vanishes but are different in general \cite{CWY_sdg}.  

Moreover, the following properties also hold: 
\begin{enumerate}
\item The total angular momentum of any spacelike hypersurface of Kerr spacetime is an invariant.
\item Let $(M, g(t),k(t))$ be a solution for the vacuum Einstein equation with lapse function $N = 1 + O(r^{-1})$ and shift vector $\gamma = \gamma^{(-1)} + O(r^{-2})$, center of mass integrals $C^i(t)$ and the total angular momentum $J^i(t)$ of $(M,g(t),k(t))$ satisfy 
\begin{align}\label{p=mv}
\begin{split}
\pl_t C^i(t) = \frac{p^i}{e}\\
\pl_t J^i(t) =0
\end{split}
\end{align}
where $(e,p^i)$ is the ADM energy-momentum 4-vector.
\end{enumerate}

In his lecture given at ETH, Christodoulou proposed definitions of angular momentum and center-of-mass integral for {\it strongly} asymptotically flat initial data sets \cite[Section 3.3, 3.4]{C_lecture} with the dynamical formula \eqref{p=mv} proved. While his angular momentum agrees with ADM's, the center of mass integral is new.
\section{Rizzi's definition of angular momentum}\subsection{Double null gauge and Rizzi's definition}
Suppose the spacetime metric takes the form
\begin{align}
g = -4\Omega^2 du dv + \slashed{g}_{AB}(dx^A - b^A dv)(dx^B - b^B dv)
\end{align}
in a double null gauge where we assume $v \in (v_0, +\infty)$ and $u \in (-\infty,+\infty)$.

The $u,v$ level sets $C_u, \underline{C}_v$ are null hypersurfaces intersecting along 2-susrfaces $S_{u,v}$. Let
\begin{align*}
r(u,v) = \sqrt{ \frac{\mbox{Area}(S_{u,v})}{4\pi} }
\end{align*}
be the areal radius. 

Associated to the double null gauge is a frame
\begin{align}\label{frame}
e_A = \frac{\pl}{\pl x^A}, \quad e_3 = \Omega^{-1}\frac{\pl}{\pl u}, \quad e_4 = \Omega^{-1}\lt( \frac{\pl}{\pl v} + b^A \frac{\pl}{\pl x^A} \rt).
\end{align} We introduce the Ricci coefficients
\begin{align*}
\chi_{AB} &= g( D_A e_4, \pl_B ), \quad
\chib_{AB} =g( D_A e_3, \pl_B ),\\
\zeta_A &= \frac{1}{2} g( D_A e_4, e_3 ),
\end{align*}
the curvature components
\begin{align*}
\underline\alpha_{AB} &= R(e_3, \pl_A, e_3, \pl_B),\\
\beta_A &= \frac{1}{2} R(\pl_A, e_4, e_3, e_4),\quad 
\underline\beta_A =\frac{1}{2} R(\pl_A, e_3, e_3, e_4),\\ 
\rho &= \frac{1}{4}R(e_3,e_4,e_3,e_4), \quad \sigma = \frac{1}{4} \mbox{}^*R(e_3,e_4,e_3,e_4),
\end{align*}
and
\begin{align*}
\widehat{\chi}_{AB} = \chi_{AB} - \frac{1}{2} (\tr_\slashed{g}\chi) \slashed{g}_{AB}, \quad
\widehat{\chib}_{AB} = \chib_{AB} - \frac{1}{2} (\tr_\slashed{g}\chib) \slashed{g}_{AB}.
\end{align*}

We also take the salient points established in \cite{CK} as assumptions: along each null hypersurface $C_u$, the limits
\begin{align}
\lim_{r\rw\infty} \hat\chi_{AB} &= \Sigma_{AB}, \quad 
\lim_{r\rw\infty} r^{-1}\hat\chib_{AB} = \Xi_{AB}, \label{def_Sigma_Xi}\\
\lim_{r\rw\infty} r^{-1} \underline\alpha_{AB} &= \underline{A}_{AB}, \quad
\lim_{r\rw\infty} r\underline\beta_A = \underline{B}_A,\label{def_A_B}\\
\lim_{r\rw\infty} r^3 \rho &= P, \label{def_P}\\ 
\lim_{r\rw\infty} r^3 \sigma &= Q 
\end{align}
exist.
 
Let $Y_{(i)}, i=1,2,3$ be the rotation Killing vector fields on $S^2$, \cite[Figure 3]{Rizzi_PRL}. Rizzi's definition of angular momentum reads \cite[(2)]{Rizzi_PRL}
\begin{align}
L(Y_{(i)}) = \frac{1}{8\pi} \lim_{v \rw \infty} \int_{S_{u,v}} \zeta_A Y^A_{(i)} d S_{\slashed g}.
\end{align}
Motivated by the Kerr spacetime, Rizzi assumes that
\begin{align}\label{assumption_beta}
\lim_{r \rw \infty} r^3\beta_A = I_A
\end{align}
exists along each null hypersurface $C_u$. By the propagation equation of $\zeta$, one obtains \cite[(3)]{Rizzi_PRL}
\begin{align}\label{Rizzi_def_beta}
L(Y_{(i)}) = \frac{1}{8\pi} \int_{S^2} \lt( \Sigma_{AB} \na_C \Sigma^{CB} + I_A \rt) Y^A_{(i)}. 
\end{align}

Besides Rizz's definition, Klainerman-Szeftel defined an angular momentum for general covariantly modulated (GCM) spheres \cite{KS} in their proof of nonlinear stability of Kerr blackhole for $|a| \ll m$.

\subsection{A conservation law for angular momentum}
Assumption \eqref{assumption_beta} is not compatible with the decay obtained in \cite{CK}
\[ \beta = o(r^{-\frac{7}{2}}) \]
(the tangential vector $\pl_A$ has scale $r$). In \cite{ChristoMG9}, Christodoulou observed  that the Bianchi identity implies that
\begin{align*}
R = \lim_{C_u^+, r \rw \infty} r^4 \underline{D}\beta
 \end{align*}
 exists. If one makes a physically reasonable assumption
\[ \lim_{u \rw -\infty} uR = R^- \neq 0,\]
then Christodoulou derived that \cite[(5)]{ChristoMG9}
\[ \beta = B_* r^{-4} \log r + B r^{-4} + o(r^{-4}) \]
uniformly in $u$ with 1-forms $B_*$ and $B$ on $S^2$ satisfying \cite[(6)]{ChristoMG9}
\begin{align*}
\frac{\pl B_*}{\pl u} &=0\\
\frac{\pl B}{\pl u} &= R.
\end{align*}

\begin{definition}
For a function $f$ on $S^2$, we denote the projection of $f$ on the sum of zeroth and first eigenspaces of $\Delta$ by $f_{[1]}$. Namely, $f_{[1]} = f_{\ell\le 1} = f_{\ell=0} + f_{\ell=1}$. For a 1-form $\omega_A = \na_A f + \epsilon_{AB} \na^B g$ on $S^2$, we denote $\omega_{A[1]} = \na_A f_{\ell=1} + \epsilon_{AB}\na^B g_{\ell=1}$.
\end{definition}

In view of \eqref{Rizzi_def_beta}, we  \cite[Section 7.1]{CKWWY_evol} extend the analysis in \cite{ChristoMG9} to get a conservation law of angular momentum
\begin{align}
\begin{split}
&\lim_{u\rw\infty} \lt( I_A - \Sigma_{AB}\na_C \Sigma^{CB} \rt)_{[1]}- \lim_{u\rw-\infty} \lt( I_A - \Sigma_{AB}\na_C \Sigma^{CB} \rt)_{[1]} \\
&= \frac{1}{2} \int_{-\infty}^\infty  -\na_A P_{\ell=1}- \epsilon_{AB} \na^B Q_{\ell=1} +\lt(  \Xi_{AB} \na_C \Sigma^{CB} - \Sigma_{AB} \na_C \Xi^{CB} \rt)_{[1]} du
\end{split}\end{align} 
where $(P,Q) = \lim_{r \rw \infty} (r^3\rho, r^3\sigma)$.

This complements the conservation law of linear momentum \cite[(13)]{Christo1991}
\[P^+_{(1)} - P^-_{(1)} = \frac{1}{4} \int_{-\infty}^{+\infty} |\Xi|^2 du, \]
which is an integrability condition for Christodoulou's memory effect.

\section{Chen-Wang-Yau angular momentum and center of mass integral at null infinity}

\subsection{Bondi-Sachs formalism}
We briefly review the Bondi-Sachs formalism and refer the readers to the excellent survey \cite{MW} for more details.

If we set $u = t-r$, the metric tensor of the Minkowski spacetime becomes
\[ - du^2 -2dudr + r^2 \sigma_{AB} dx^A dx^B. \]
Similarly, the Schwarzschild metric in Eddington-Finkelstein coordinates is
\[ - (1 - \frac{2m}{r})du^2 - 2 dudr + r^2 \sigma_{AB} dx^A dx^B \]

Taking the above two examples as models, Bondi and his collaborators postulated that the spacetime admits a coordinate system (Bondi-Sachs coordinates) $(u,r,x^A)$ where $u$ is an optical function, $r$ is the ``luminosity distance" from the source, and $x^A, A=2,3$ are the coordinates of the spherical section; $r \in (r_0, \infty), u \in (u_0,u_1)$. In Bondi-Sachs coordinates, the metric tensor takes the form
\begin{align}\label{BS}
-UV du^2 - 2 U dudr + r^2 h_{AB}(dx^A + W^A du)(dx^B + W^B du).
\end{align}

Assuming the spacetime is asymptotically Minkowskian, the metric coefficients satisfy \begin{align}\label{boundary condition}
U \rw 1, V \rw 1, h_{AB} \rw \sigma_{AB}, W^A \rw 0
\end{align} as $r \rw \infty$.

Bondi and his collaborators make two more assumptions. The first is a determinant condition $\det (h_{AB}) = \det (\sigma_{AB})$. The second is the ``outgoing radiation condition" that all metric coefficients can be expanded into power series of $\frac{1}{r}$ with coefficients being functions of $u, x^A$.

The null vacuum Einstein constraint equations then enjoy a remarkable hierarchy and all metric coefficients can be determined term-by-term
\[\begin{split} U&=1 - \frac{1}{16r^2} |C|^2 + O(r^{-3}),\\
V&=1-\frac{2m}{r}+ \frac{1}{r^2}\lt( \frac{1}{3}\na^A N_A + \frac{1}{4} \na^A C_{AB} \na_D C^{BD} + \frac{1}{16}|C|^2 \rt) + O(r^{-3}),\\
W^A&= \frac{1}{2r^2} \na_B C^{AB} + \frac{1}{r^3} \lt( \frac{2}{3}N^A - \frac{1}{16} \na^A |C|^2 -\frac{1}{2} C^{AB} \na^D C_{BD} \rt) + O(r^{-4}),\\
h_{AB}&={\sigma}_{AB}+\frac{C_{AB}}{r}+ \frac{1}{4r^2} |C|^2 \sigma_{AB} + O(r^{-3})\end{split} \] where  $m=m(u, x^A)$ is the mass aspect, $N_A = N_A(u,x^A)$ is the angular aspect and $C_{AB}=C_{AB}(u, x^A)$ is the shear tensor of this Bondi-Sachs coordinate system. 

Bondi-Sachs extract physical information from the coefficients of expansion; the Bondi-Sachs energy momentum 4-vector $(E,P^k)$  for a $u = const.$ hypersurface is given by
\begin{equation}\begin{split} E(u) &= \frac{1}{4\pi}\int_{S^2} m\\
 P^k(u) &= \frac{1}{4\pi} \int_{S^2} m \tilde{X}^k, k=1, 2, 3\end{split}\end{equation}
where $\tilde{X}^k, k=1, 2 ,3$ are the standard coordinate functions on $\R^3$ restricted to the unit sphere $S^2$. 

From the vacuum Einstein equations,
\begin{align}\label{evolution of mass aspect}
\pl_u m = -\frac{1}{8}N_{AB}N^{AB} + \frac{1}{4} \na^A\na^B N_{AB},
\end{align}
one obtains the famous Bondi mass loss formula
\begin{align}
\pl_u E = -\frac{1}{32\pi} \int_{S^2} N_{AB}N^{AB}.\label{mass loss}
\end{align}
Bondi-Sachs formalism is a coordinate-based approach and one should take all coordinate transformations that preserve the form of the metric tensor \eqref{BS} and the asymptotic conditions \eqref{boundary condition} into consideration. One such coordinate transformation is {\it supertranslation}: any smooth function $f(x)$ on $S^2$ gives rise to a coordinate transformation
\begin{align}\label{supertranslation}
\begin{split}
\bar u &= u + f(x) + \frac{\bar u^{(-1)}(u,x)}{r} + o(r^{-1}),\\
\bar r &= r + \bar r^{(0)}(u,x) + o(1),\\
\bar x^A &= x^A + \frac{\bar x^{A(-1)}(u,x)}{r} + o(r^{-1})
\end{split}
\end{align}
where $\bar u^{(-1)}, \bar r^{(0)}, \bar x^{A(-1)}$ are determined by $f$ \cite{BVM}. 

To end this subsection, we note that the limits of the Ricci coefficients with respect to null frame
\[ e_4 = \pl_r, \quad e_3 = \frac{2}{U}\lt( \pl_u - W^D \pl_D - \frac{V}{2}\pl_r \rt)\]
are given by \cite[Appendix A]{CKWWY_evol}
\begin{align}
\begin{split}\label{shear, news}
\Sigma_{AB} &= -\frac{1}{2}C_{AB},\\
\Xi_{AB} &= N_{AB}.
\end{split}
\end{align} 

\subsection{Supertranslation in double null gauge}
Suppose the spacetime metric takes the form
\begin{align}
g = -4\Omega^2 du dv + \slashed{g}_{AB}(dx^A - b^A dv)(dx^B - b^B dv)
\end{align}
in a double null gauge where we assume $v \in (v_0, +\infty)$ and $u \in (-\infty,+\infty)$.
We assume that
\begin{align}
\Omega &= 1 + o(1), \quad \pl_A\Omega = o(1) \label{AF,Omega}\\
\slashed{g}_{AB} &= r^2 \sigma_{AB} + O(r)\\
b^A &= O(r^{-2})\label{AF,b}
\end{align}
along each outgoing null hypersurface $C_u$. Moreover, assume \eqref{def_Sigma_Xi}, \eqref{def_A_B}, \eqref{def_P} hold.
\begin{remark}
Since $\frac{\pl b^A}{\pl u} = 4\Omega^2 (\slashed{g}^{-1})^{AB}\zeta_B$ \cite[(1.199)]{C_formation},  assumption \eqref{AF,b} is satisfied if it is satisfied on some $C_u$ and the decay $\zeta_A = O(r^{-1})$ holds. 
\end{remark}

Given a change of double null gauge 
\begin{align*}
u = f(\tilde u, \tilde x), \quad
v = \tilde v, \quad
x^A = \tilde x^A,
\end{align*}
the metric tensor is transformed into
\[ g = -4\Omega^2 \frac{\pl f}{\pl\tilde u} d\tilde u d\tilde v - 4\Omega^2 \frac{\pl f}{\pl\tilde x^A} d\tilde x^A d\tilde v + \slashed{g}_{AB}(d\tilde x^A - b^A d\tilde v)(d\tilde x^B - b^B d\tilde v) \]
and remains in the double null gauge. In order to preserve \eqref{AF,Omega}, it is necessary $\frac{\pl f}{\pl \tilde u} = 1$ and hence the two optical functions are related by
\begin{align}
u = \tilde u + \frac{f(x)}{2} \label{str_double null}
\end{align}
for some function $f$ on $S^2$. The discrepancy of \eqref{str_double null} and \eqref{supertranslation} comes from the fact the optical function $u$ used in double null gauge is equal to $\frac{t-r}{2}$ in Minkowski spacetime. From now on we identify $\tilde v, \tilde x^A$ with $v, x^A$ and call \eqref{str_double null} supertranslations in double null gauge.

We compute how the Ricci coefficients transform. 
\begin{prop}
Under a supertranslation in double null gauge $u = \tilde u + \frac{f(x)}{2}$, the limits of traceless second fundamental forms transform as
\begin{align}
\tilde\Sigma_{AB}( \tilde u, x) &= \Sigma_{AB} \lt( \tilde u + \frac{f(x)}{2}, x \rt) + \na_A\na_B f - \frac{1}{2}\Delta f \,\sigma_{AB}, \label{str_Sigma}\\
\tilde\Xi_{AB}(\tilde u, x) &= \Xi_{AB}\lt( \tilde u + \frac{f(x)}{2}, x \rt) \label{str_Xi},
\end{align}
and the limit of curvature component $\rho$ \eqref{def_P} transform as
\begin{align}\label{str_P}
\tilde P(\tilde u, x) = P -\na^A f \underline{B}_A  + \frac{1}{4} \na^A f\na^B f \underline{A}_{AB}
\end{align}
where the the curvature components on right-hand side are evaluated at $\lt(\tilde u + \frac{f}{2}, x \rt)$.
\end{prop}
\begin{proof}
We compute
\begin{align*}
du &= d\tilde u + \frac{\pl_A f}{2} d\tilde x^A\\
dv &= d\tilde v\\
dx^A &= d\tilde x^A.\\
\end{align*}
Thus the coordinate vector fields transform as
\begin{align*}
\frac{\pl}{\pl \tilde u} &= \frac{\pl}{\pl u}\\
\frac{\pl}{\pl \tilde v} &= \frac{\pl}{\pl v}\\
\frac{\pl}{\pl \tilde x^A} &= \frac{\pl_A f}{2}\frac{\pl}{\pl u} + \frac{\pl}{\pl x^A}
\end{align*}
and the frame \eqref{frame} transforms as
\begin{align*}
\tilde e_A &= \frac{\pl_A f}{2}\Omega e_3 + e_A\\
\tilde e_3 &= \tilde\Omega^{-1}\Omega e_3\\
\tilde e_4 &= \tilde\Omega^{-1} \lt( \Omega e_4 - b^A e_A + \tilde b^A \tilde e_A \rt)\\
&= \tilde\Omega^{-1} \lt( \Omega e_4 + \lt( \tilde b^A - b^A \rt) e_A + \tilde b^A \frac{\pl_A f}{2}\Omega e_3 \rt).
\end{align*}

Recall the identification $\tilde v, \tilde x^A$ with $v, x^A$. The metric coefficients are related by
\begin{align}\label{transformation_metric}
\begin{split}
\tilde\Omega &= \Omega\\
\tilde g_{AB} &= g_{AB}\\
\tilde g_{AB} \tilde b^B &= \Omega^2 \pl_A f + g_{AB}b^B
\end{split}
\end{align}
where the left-hand side are evaluated at $(\tilde u,v,x)$ and the right-hand side are evaluated at $\lt( \tilde u + \frac{f}{2}, v, x \rt)$.
 
We compute the transformation of $\chib$
\begin{align*}
\lt\langle D_{\tilde e_A} \tilde e_3, \tilde e_B \rt\rangle = \tilde\Omega^{-1}\Omega \lt\langle D_{e_A + \frac{1}{2}\pl_A f \Omega e_3} e_3, e_B \rt\rangle = \tilde\Omega^{-1}\Omega \lt\langle D_{e_A} e_3, e_B \rt\rangle
\end{align*}
and take limit to get
\begin{align*}
\tilde \Xi_{AB} = \Xi_{AB}.
\end{align*}

We compute the transformation of $\chi$
\begin{align*}
\lt\langle D_{\tilde e_A} \tilde e_B, \tilde\Omega \tilde e_4 \rt\rangle &= \pl_A \lt( \frac{\pl_B f}{2}\Omega \rt)\langle e_3, \Omega e_4 \rangle \\
&\quad+ \frac{\pl_B f}{2} \Omega \lt\langle D_{e_A + \frac{\pl_A f}{2} \Omega e_3} e_3, \Omega e_4 + \lt( \tilde b^C - b^C \rt)e_C \rt\rangle \\
&\quad + \lt\langle D_{e_A + \frac{\pl_A f}{2} \Omega  e_3} e_B, \Omega e_4 + \lt( \tilde b^C - b^C \rt)e_C + \tilde b^A \frac{\pl_A f}{2} \Omega e_3 \rt\rangle\\
&= - \pl_A \pl_B f + \lt\langle D_{e_A} e_B, \Omega e_4 \rt\rangle + \langle D_{e_A} e_B, e_C \rangle \lt( \tilde b^C - b^C \rt) + O(r^{-1}).
\end{align*}
The limit of the traceless part is thus
\begin{align*}
\tilde\Sigma_{AB} = \Sigma_{AB} + \na_A\na_B f - \frac{1}{2} \Delta f \,\sigma_{AB}.
\end{align*}

Finally, we compute the transformation of $\rho$
\begin{align*}
\tilde\rho &= \frac{1}{4}R \lt( \tilde e_3,\tilde e_4,\tilde e_3, \tilde e_4 \rt) \\ 
&= \frac{1}{4} R(e_3,e_4,e_3,e_4) + \frac{1}{2} ( \tilde b^A\ - b^A ) R(e_3,e_4,e_3,e_A) \\
&\quad + \frac{1}{4} ( \tilde b^A - b^A )( \tilde b^B - b^B ) R(e_3,e_A,e_3,e_B) + O(r^{-4})\\
&= \rho - \frac{1}{r^2} \na^A f \underline\beta_A + \frac{1}{4r^4} \na^A f \na^B f \underline\alpha_{AB} + O(r^{-4})
\end{align*}
and take limit. 
\end{proof}

The supertranslation can be used to kill the closed part\footnote{It is called the electric part in physics literature.} of the shear tensor $\Sigma$.  
\begin{theorem}
Suppose
\begin{align}
\sup_{(-\infty,\infty) \times S^2} \Xi_{AB}\Xi^{AB} < 48. \label{assumption_48}
\end{align}
Then for a fixed $\tilde u = \tilde u_0$ and any $\psi_{sl} \in \mathcal{H}_{\ell\le 1},$ there exists a unique $\Psi \in \mathcal{H}_{\ell\ge 2}$ such that $\tilde\Sigma_{AB}(\tilde u) = \Sigma_{AB}(\tilde u + \frac{\psi(x)}{2},x) + \na_A\na_B \psi - \frac{1}{2}\Delta\psi \sigma_{AB}$ with $\psi = \psi_{sl} + \Psi$ has no closed part.
\end{theorem}
Since the vector space $\mathcal{H}_{\ell\le 1}$ is 4-dimensional, the result singles out a cut, with a 4-dimensional degree of freedom, of null infinity. In Section 4.2 of \cite{Rizzi_thesis}, Rizzi considers the lapse transformation $L \rw a^{-1}L, \Lb \rw a \Lb$ with $\lim_{r\rw\infty} a = \psi$ along each null hypersurface. This leads to a linear equation
\[ \Delta(\Delta + 2)\psi = \na^A\na^B (\psi \Xi_{AB}), \]
which is the linearized equation of \eqref{no closed equation}, whose solution is used to describe a procedure that retains the 4-dimensional dependence of origins in his definition. See Chapter 4 and 5 of \cite{Rizzi_thesis}. The original proof of Rizzi's version \cite[Theorem 1]{Rizzi_thesis}, due to Christodoulou, requires the upper bound 16. We follow his argument with improved estimates.
\begin{proof}
The equation to be solved is
\begin{align}\label{no closed equation}
\Delta(\Delta + 2) \Psi = -2\na^A\na^B \lt( \Sigma_{AB}(\tilde u + \frac{\Psi + \psi_{sl}}{2},x)  \rt).
\end{align}
Set up an iterative equation
\[ \Delta(\Delta + 2) \Psi_{n+1} = -2\na^A\na^B \lt( \Sigma_{AB}( \tilde u + \frac{\Psi_n + \psi_{sl}}{2},x ) \rt).\]
and start with $\Psi_0 =0$. We will show that $\Psi_n$ converge to $\Psi$ in $L^2$ and this proves the existence.

Let $h_{n+1} = \Psi_{n+1} - \Psi_n$ and we have
\begin{align*}
\Delta(\Delta + 2)h_{n+1} = -2\na^A\na^B \lt( \Sigma_{AB}(\tilde u + \frac{\Psi_n + \psi_{sl}}{2},x) - \Sigma_{AB}(\tilde u + \frac{\Psi_{n-1} + \psi_{sl}}{2},x) \rt).
\end{align*} 

Let us analyze the equation
\begin{align}\label{temp_g}
\Delta(\Delta + 2)g = \na^A\na^B \eta_{AB}
\end{align}
for $g$ supported in $\ell\ge 2$ modes and a symmetric traceless 2-tensor $\eta_{AB}$. Multiplying by $g$ and integrating, we get
\[ \int_{S^2} \Delta g (\Delta+2)g = \int_{S^2} \eta_{AB} \lt( \na^A\na^B g - \frac{1}{2}\Delta g \sigma^{AB} \rt).\]
Decompose $g$ into spherical harmonics $g = \sum_{\ell =2}^\infty g_l$. Because of the orthogonality of $g_l$'s, we have
\begin{align}
\int_{S^2} \Delta g (\Delta+2)g = \sum_{\ell=2}^\infty
 [-\ell(\ell+1)][2-\ell(\ell+1)] \int_{S^2} g_\ell^2 \ge 24\int_{S^2} g^2 \label{1}.
\end{align}
On the other hand, for any smooth function $f$ on $S^2$, integration by parts yields the identity 
\[ \int_{S^2} \lt( \na_A\na_B f - \frac{1}{2}\Delta f \sigma_{AB}\rt)\lt( \na^A\na^B f - \frac{1}{2}\Delta f\sigma^{AB}\rt) = \frac{1}{2} \int_{S^2} \Delta f(\Delta+2) f. \]
By Cauchy-Schwarz and H\"{o}lder inequality, we get
\begin{align*}
\int_{S^2} \eta_{AB} \lt( \na^A\na^B g - \frac{1}{2}\Delta g \sigma^{AB} \rt) \le \sqrt{ \int_{S^2} |\eta|^2  \cdot \int_{S^2} |\na^2 g - \frac{1}{2}(\Delta g)\sigma|^2 }. 
\end{align*}
Putting these together, we obtain
\begin{align*}
24 \int_{S^2} g^2 \le \frac{1}{2} \int_{S^2} |\eta|^2.
\end{align*}

By the Mean Value Theorem, for each $x\in S^2$ \[  \lt| \Sigma_{AB}(\tilde u + \frac{\Psi_n + \psi_{sl}}{2}(x),x) - \Sigma_{AB}(\tilde u + \frac{\Psi_{n-1} + \psi_{sl}}{2}(x),x) \rt| \le  \sup_{(-\infty,\infty) \times \{x\} } |\Xi|  \cdot \frac{1}{2} |\Psi_{n} - \Psi_{n-1}|(x), \]
and integrating over $S^2$ we get
\[ \int_{S^2} h_{n+1}^2 \le \frac{1}{48} \sup_{(-\infty,\infty) \times S^2} |\Xi|^2 \int_{S^2} h_n^2. \]
If $\sup |\Xi|^2 < 48$, then $h_n$ converge to $0$ and hence $\Psi_n$ converge to $\Psi$ in $L^2$.

The uniqueness follows similarly. Suppose $\phi_1,\phi_2$ solve the equation with the same $\psi_{sl}$. Then $\phi = \phi_2 - \phi_1$ solves $\Delta(\Delta + 2)\phi = -2\na^A\na^B (\Sigma_{AB}(
\tilde u + \frac{\psi_{sl}+\phi_2}{2},x) - \Sigma_{AB}(\tilde u + \frac{\psi_{sl} + \phi_1}{2},x))$. The above analysis implies that
\[ \int_{S^2} \phi^2 \le \frac{1}{48} \sup |\Xi|^2 \int_{S^2} \phi^2 \]
and hence $\phi=0$ by \eqref{assumption_48}.
\end{proof}

\subsection{Total flux of angular momentum and the supertranslation ambiguity}
The time derivative of Rizzi's angular momentum is given by \cite[(4)]{Rizzi_PRL}
\begin{align*}
\frac{\pl L(Y_{(k)})}{\pl u} = \frac{1}{4\pi} \int_{S^2} \lt\{ -\Xi_{AB} \na_C \Sigma^{CB} + \frac{1}{2}\lt( \Sigma_{AB}\na_C \Xi^{CB} - \Sigma^C_B\na^B \Xi_{CA} \rt) \rt\} Y^A_{(k)}.
\end{align*}
\begin{remark}\label{difference in partial u}
Note that the factor is $\frac{1}{4\pi}$ rather than $\frac{1}{8\pi}$. The reason is that the parameter $u$ used in \cite{Christo1991} or \cite{Rizzi_PRL} (and used in Bondi-Sachs formalism) is equal to $t-r$ plus some constant in Minkowski spacetime. When citing their formula, the derivative $2 \frac{\pl}{\pl u}$ should be replaced by $\frac{\pl}{\pl u}$. For one more example, equation (5) of \cite{Christo1991} now reads 
\begin{align}\label{dSigma=Xi}
\frac{\pl\Sigma}{\pl u} = -\Xi.
\end{align} 
\end{remark}

For the rest of the paper, we assume 
\begin{align}\label{assumption_Xi}
\Xi_{AB} = O(|u|^{-1-\varepsilon})
\end{align}
for some $\varepsilon > 0$ as $u$ approaches $\pm\infty$.

We fix the rotation Killing vector $Y^A = \epsilon^{AB}\na_B \tilde X^k$ for some $k \in \{ 1,2,3\}$ for definiteness. Integrating by parts the last term and then integrating from $u=-\infty$ to $u = + \infty$, we obtain the {\it total flux} of Rizzi's angular momentum
\begin{align*}
\delta L = \frac{1}{8\pi} \int_{-\infty}^{+\infty} \lt( \int_{S^2}  Y^A \lt( \Sigma_{AB}\na_C \Xi^{CB} -\Xi_{AB} \na_C \Sigma^{CB}  \rt) + \epsilon^{AB}\tilde X^k \lt( \Sigma_A^C \Xi_{CB} \rt) \rt) du.
\end{align*}
The integral is finite because of \eqref{assumption_Xi}. In view of \eqref{shear, news}, the result, up to a minus sign\footnote{Firstly, the factor $\frac{1}{8\pi}$ is omitted in \cite{CKWWY_evol}. Secondly, take into account the discrepancy in $u$. Lastly, for Kerr spacetime Rizzi's definition gives $ma$ while Chen-Wang-Yau gives $-ma$.}, coincides with the classical Ashtekhar-Streubel flux of angular momentum \cite{AS}, see \cite[Theorem 1.2]{CKWWY_evol}.

\begin{definition}
Define 
\begin{align}
2m(u,x) = \lim_{r \rw\infty} r^3 \lt( K + \frac{1}{4} \tr\chi \tr\chib \rt),
\end{align}
the energy $E(u) = \frac{1}{4\pi} \int_{S^2} m$, and the linear momentum $P^k(u) = \frac{1}{4\pi} \int_{S^2 } \tilde X^k m$. Here $K$ is the Gauss curvature of $S_{u,v}$.
\end{definition}

The total flux of angular momentum changes when computing in different double null gauges. 
\begin{theorem}
Consider a supertranslation in a double null gauge $u = \tilde u + \frac{f(x)}{2}$ and let $\delta L_f$ denote the total flux of angular momentum computed in $(\tilde u, v,x)$ gauge. Suppose \eqref{assumption_Xi} holds, then 
\begin{align}
\delta L_f - \delta L = \frac{1}{4\pi} \int_{S^2} f Y^A \na_A \lt( m^+ - m^- \rt) 
\end{align}
where
\[ m^{\pm} = \lim_{u\rw \pm\infty} m(u,x). \]
\end{theorem}
\begin{proof}
We sketch the proof. For details, see Section 5.1 of \cite{CKWWY_evol}. First rewrite $\delta L$ as
\begin{align*}
\frac{1}{8\pi} \int_{-\infty}^{+\infty} \lt( \int_{S^2} -\na_C Y^A \Sigma_{AB} \Xi^{CB} - Y^A \na_C \Sigma_{AB} \Xi^{CB} - Y^A \Xi_{AB} \na_C \Sigma^{CB} + \epsilon^{AB}\tilde X^k \Sigma_A^C \Xi_{CB} \rt) du
\end{align*}
For $\delta L_f$, simply replace $\Sigma, \Xi, du$ by $\tilde\Sigma, \tilde\Xi, d\tilde u$. Applying the chain rule to \eqref{str_Sigma} yields
\begin{align}\label{chain rule Sigma}
\begin{split}
\na_C\tilde\Sigma_{AB}(\tilde u, x) &= -\Xi_{AB}( \tilde u + \frac{f}{2}, x ) \frac{\na_C f}{2} + ( \na_C \Sigma_{AB} )( \tilde u + \frac{f}{2},x ) + \na_C ( \na_A\na_B f - \frac{1}{2}\Delta f \sigma_{AB} ),\\
\na_C \tilde\Sigma^{CB}(\tilde u, x) &= - \Xi^{CB}( \tilde u + \frac{f}{2},x)\frac{\na_C f}{2} + \lt( \na_C \Sigma^{CB}\rt) ( \tilde u + \frac{f}{2},x ) + \frac{1}{2} \na^B (\Delta + 2)f,
\end{split} 
\end{align}
where we used \eqref{dSigma=Xi}.

After a series of change of variables and integration by parts, the difference of total flux is put into the following simple form
\begin{align*}
\delta L_f - \delta L = -\frac{1}{16\pi} \int_{-\infty}^{+\infty} \lt( \int_{S^2} f Y^A  \na_A \lt( |\Xi|^2 - 2 \na^B\na^C \Xi_{BC} \rt) \rt) du.
\end{align*}

On the other hand, equations (1), (3), (4), (7) of \cite{Christo1991} together imply
\begin{align*}
-\na^B\na^C \Xi_{BC} = -\frac{1}{2}|\Xi|^2 + \frac{\pl}{\pl u} \lt( P - \frac{1}{2} \Sigma_{AB}\Xi^{AB} \rt).
\end{align*}
Moreover, by (3) of \cite{Christo1991}, we have
\begin{align}
2m = - P + \frac{1}{2} \Sigma_{AB} \Xi^{AB} \label{m and P}
\end{align}
and thus
\begin{align}\label{evolution of mass aspect'}
\pl_u m = -\frac{1}{4}|\Xi|^2 + \frac{1}{2} \na^B\na^C\Xi_{BC}.
\end{align}
\end{proof}

When $f$ is supported in $\ell \le 1$ modes, the difference is equal to the total flux of linear momentums. However, if $f$ has $\ell \ge 2$ modes, the formula bears no resemblance to \eqref{SR} and physicists have yet to find an interpretation. This is the so-called ``supertranslation ambiguity". For more on supertranslation ambiguity, we refer the readers to Ashtekar-De Lorenzo-Khera \cite{ADK}. 

\subsection{Supertranslation invariance of the Chen-Wang-Yau angular momentum and center of mass integral}
The total flux of the Chen-Wang-Yau angular momentum $\delta J$ is given by \cite{CKWWY_evol}
\begin{align*}
\delta J = \delta L - \frac{1}{8\pi} \lt[ \int_{S^2} Y^A \mathfrak{s} \na_A m \rt]_{-\infty}^{+\infty}
\end{align*}
where the function $\mathfrak{s}(u,x)$ is the potential in the Hodge decomposition
\begin{align*}
\Sigma_{AB} = \na_A\na_B \mathfrak{s} - \frac{1}{2}\Delta \mathfrak{s} \,\sigma_{AB} + (\mbox{co-closed part}).
\end{align*}
We assume $\mathfrak{s}$ is supported in $\ell \ge 2$ modes and is thus unique. This nonlocal term $\mathfrak{s}$ appears when solving the optimal isometric embedding equation \cite{KWY}. The main result in \cite{CKWWY_evol, CWWY_atmp} is that the total flux of the Chen-Wang-Yau angular momentum transforms according to \eqref{SR} under supertranslations. 
\begin{theorem}
Consider a supertranslation in double null gauge $u = \tilde u + \frac{f(x)}{2}$ and let $\delta J_f$ denote the total flux of the Chen-Wang-Yau angular momentum computed in $(\tilde u, v,x)$ gauge. Suppose \eqref{assumption_Xi} holds, then
\begin{align}
\delta J_f - \delta J = -\varepsilon^{ik}_{\;\;\;\;l}\alpha_i \delta P^l
\end{align}
where $f_{\ell\le 1} = \alpha_0 + \alpha_i \tilde X^i,$ $\varepsilon$ is the Levi-Civita symbol, and $\delta P^l = \frac{1}{4\pi} \int_{S^2} \tilde X^l(m^+ - m^-)$ is the total flux of linear momentum.
\end{theorem}
\begin{proof}
By \eqref{str_Sigma}, we have
\begin{align*}
\tilde{\mathfrak{s}}^{\pm}(x) = \mathfrak{s}^\pm(x) + f_{\ell\ge 2}.
\end{align*}
Moreover, by virtue of equations (2) and (6) of \cite{Christo1991}, \eqref{str_P} and the assumption on the decay of $\Xi$ \eqref{assumption_Xi} imply $\tilde P^{\pm} = P^{\pm}$. Consequently, 
\begin{align*}
\tilde m^{\pm}(x) = m^{\pm}(x)
\end{align*}
by \eqref{m and P}.

Putting these together, we obtain 
\[\delta J_f - \delta J = \frac{1}{8\pi} \int_{S^2} f_{\ell\le 1} Y^A \na_A (m^+ - m^-).\]
Recall that $Y^A = \epsilon^{AB}\na_B \tilde X^k$ and we have
\begin{align*}
Y^A \na_A f_{\ell\le 1} = \varepsilon^{ik}_{\;\;\;\;l} \alpha_i \tilde X^l.
\end{align*}
\end{proof}

Next, we turn to the effect of supertranslation on the center of mass integral. In double null gauge, the total flux of Ashtekar-Streubel center of mass integral $\delta\tilde C^k$ is given by 
\begin{align*}
\delta\tilde C^k =\frac{1}{8\pi} \int_{-\infty}^{+\infty} \lt[ \int_{S^2} \nabla^{A}\tilde{X}^{k} \lt(  -\frac{u}{2} \nabla_A |\Xi|^2+ \Sigma_{AB}\nabla_{D}\Xi^{BD} - \Xi_{AB}\nabla_{D}\Sigma^{BD} \rt)  \rt] du,
\end{align*}
up to a minus sign. Again, compare \cite[Theorem 1.2]{CKWWY_evol}.
\begin{theorem}
Consider a supertranslation in double null gauge $u = \tilde u + \frac{f(x)}{2}$ and let $\delta \tilde C_f$ denote the total flux of Ashtekar-Streubel center of mass integral computed in $(\tilde u, v,x)$ gauge. Suppose \eqref{assumption_Xi} holds, then \begin{align*}
\delta\tilde C^k_f - \delta\tilde C^k = \frac{1}{8\pi} \int_{S^2} \lt( -6\tilde X^k (m^+ - m^-) + 2 \na^A\tilde X^k \na_A(m^+ - m^-) \rt)f.
\end{align*}
\end{theorem}
\begin{proof}
First of all, note that the first term of $\delta\tilde C^k$ is equal to $\int_{S^2} -u \tilde X^k |\Xi|^2$. For $\delta \tilde C_f$, simply replace $\Sigma, \Xi, du$ by $\tilde\Sigma, \tilde\Xi, d\tilde u$. Applying the chain rule to \eqref{str_Xi}, we get 
\begin{align}
\na_D \tilde\Xi^{BD}(\tilde u, x) = \pl_u\Xi^{BD}\lt( \tilde u + \frac{f(x)}{2}, x \rt) \frac{\na_D f}{2} + \na_D \Xi^{BD}\lt( \tilde u + \frac{f(x)}{2}, x \rt).
\end{align}
Combining with \eqref{chain rule Sigma}, we obtain
\begin{align*}
&\delta\tilde C^k_f - \delta\tilde C^k \\
&= \frac{1}{8\pi}\int_{-\infty}^{+\infty} \lt[ \int_{S^2} \frac{f}{2}\tilde X^k |\Xi|^2 + \na^A\tilde X^k (\na_A\na_B f - \frac{1}{2}\Delta f \sigma_{AB})(\pl_u \Xi^{BD} \frac{\na_D f}{2} + \na_D \Xi^{BD})\rt]du \\
&\quad + \frac{1}{8\pi}\int_{-\infty}^{+\infty} \lt[ \int_{S^2} \na^A \tilde X^k \lt( \Sigma_{AB} \pl_u \Xi^{BD} \frac{\na_D f}{2} - \Xi_{AB} ( -\Xi^{CB}\frac{\na_C f}{2} + \frac{1}{2}\na^B(\Delta+2)f )\rt)\rt] du\\
&= \frac{1}{8\pi} \int_{-\infty}^{+\infty} \lt[ \int_{S^2} \frac{f}{2}\tilde X^k |\Xi|^2 + \na^A \tilde X^k (\na_A\na_B f - \frac{1}{2}\Delta f\sigma_{AB}) \na_D \Xi^{BD} \rt]du\\
&\quad + \frac{1}{8\pi}\int_{-\infty}^{+\infty} \lt[ \int_{S^2} \na^A \tilde X^k \lt( |\Xi^2| \frac{\na_A f}{2} - \frac{1}{2} \Xi_{AB}\na^B (\Delta + 2)f \rt) \rt] du
\end{align*}
where we integrate by parts in $u$ and use \eqref{dSigma=Xi} together with the identity $\Xi_{AB}\Xi^{BC} = \frac{1}{2}\delta_A^C |\Xi|^2$ in the second equality. Further integration by parts on $S^2$ (see the proof of Theorem 5.3 \cite{CKWWY_evol}) leads to
\begin{align*}
&\delta\tilde C^k_f - \delta\tilde C^k \\
&= \frac{1}{8\pi}\int_{-\infty}^{+\infty} \lt[ \int_{S^2} f\tilde X^k\lt( \frac{|\Xi|^2}{2} - \na^A\na^B\Xi_{AB} \rt) + \na^A \tilde X^k \lt( \frac{|\Xi|^2}{2} - \na^B\na^C\Xi_{BC} \rt) \na_A f \rt] du\\
&= \frac{1}{8\pi} \int_{S^2} \lt( -6 \tilde X^k (m^+ - m^-) + 2 \na^A\tilde X^k \na_A(m^+ - m^-) \rt)f 
\end{align*}
\end{proof}

On the other hand, the total flux of the Chen-Wang-Yau center of mass-integral transforms according to \eqref{SR} under supertranslations.
\begin{theorem}
Consider a supertranslation in a double null gauge $u = \tilde u + \frac{f(x)}{2}$ and let $\delta C_f$ denote the total flux of the Chen-Wang-Yau center of mass integral computed in $(\tilde u, v,x)$ gauge. Suppose \eqref{assumption_Xi} holds,\begin{align}
\delta C^k_f - \delta C^k = -\alpha_0 \delta P^k - \alpha_k \delta E
\end{align}
where $f_{\ell\le 1} = \alpha_0 + \alpha_i \tilde X^i,$ and $\delta E = \frac{1}{4\pi} \int_{S^2} m^+ - m^-$ is the total flux of energy.
\end{theorem}
\begin{proof}
The Chen-Wang-Yau center of mass integral  has the correction term on the total flux
\begin{align*}
\delta C^k = \delta\tilde C^k + \lt. \frac{1}{8\pi} \int_{S^2} -6\tilde X^k \mathfrak{s} m  + 2 \na^A \tilde X^k \mathfrak{s} \na_A m \rt]_{-\infty}^{+\infty}.
\end{align*}
Recall that under the supertranslation, 
\[ \tilde m^{\pm}(x) = m^{\pm}(x), \quad \tilde{\mathfrak{s}}^{\pm}(x) = \mathfrak{s}^{\pm}(x) + f_{\ell\ge 2},\]
and we obtain
\begin{align*}
\delta C^k_f - \delta C^k &= \frac{1}{8\pi} \int_{S^2} \lt( -6\tilde X^k (m^+ - m^-) + 2 \na^A\tilde X^k \na_A(m^+ - m^-) \rt) f_{\ell\le 1}\\
&= \frac{1}{8\pi} \int_{S^2} -2 \alpha_0 \tilde X^k (m^+ - m^-) -2  \alpha_k (m^+ - m^-),
\end{align*}
where we used the identity $\na^A \tilde X^k \na_A \tilde X^i = \delta^{ik} - \tilde X^k\tilde X^i$.
\end{proof} 

We remark that for non-radiative spacetimes, namely $\Xi \equiv 0$, the Chen-Wang-Yau angular momentum and center of mass integral are constant in $u$ and transform according to \eqref{SR} under supertranslations. See \cite[Theorem 6.2]{CKWWY_evol} or \cite[Theorem 2 and Section 5]{CWWY_atmp}. Moreover, the total fluxes of the Chen-Wang-Yau angular momentum and center of mass (or themselves in non-radiative spacetimes) transform equivariantly under Lorentz transformations (defined in Bondi-Sachs formalism) \cite{CWWY_BMS}.

\end{document}